\documentclass[letterpaper, 10 pt, conference]{ieeeconf}
\IEEEoverridecommandlockouts
\usepackage{amsmath}
\usepackage{amsthm}
\usepackage{amssymb}
\usepackage{mathtools}
\usepackage[utf8]{inputenc}
\usepackage[T1]{fontenc}
\usepackage{hyperref}
\usepackage{url}
\usepackage{booktabs}
\usepackage{amsfonts}
\usepackage{nicefrac}
\usepackage{microtype}
\usepackage{xcolor}
\usepackage{times}
\usepackage[pdftex]{graphicx}
\usepackage{epsfig}
\usepackage{epsfig}
\usepackage{enumerate}
\usepackage{color}
\usepackage{comment}
\usepackage{cite}
\usepackage{caption}
\usepackage{todonotes}
\usepackage{multirow}
\usepackage{diagbox}
\usepackage{mathabx}
\usepackage{subcaption}
\usepackage[subtle,bibnotes,leading,lists]{savetrees}

\newtheorem{remark}{Remark}
\newtheorem{problem}{Problem}

\newtheorem{proposition}{Proposition}

\newcommand{\frameworkname}{MA-COPP}
\newcommand{\sequence}[3]{#1_{#2\ldots#3}}

\newcommand{\ftenth}[0]{F1TENTH }

\begin{document}

\title{Conformal Off-Policy Prediction for Multi-Agent Systems}

\author{{Tom Kuipers$^\Asterisk$, Renukanandan Tumu$^\Asterisk$, Shuo Yang, Milad Kazemi, Rahul Mangharam, and Nicola Paoletti}
\thanks{T. Kuipers, M. Kazemi and N. Paoletti are with the Department of Informatics, King's College London, UK. Email: {\tt\small\{tom.kuipers, milad.kazemi, nicola.paoletti\}@kcl.ac.uk}}
\thanks{R. Tumu, S. Yang and R. Mangharam are with the Department of Electrical and Systems Engineering, University of Pennsylvania, USA. Email: {\tt\small\{nandant, yangs1, rahulm\}@seas.upenn.edu}}
\thanks{$^\Asterisk$Authors contributed equally.}
}

\maketitle

\begin{abstract}
Off-Policy Prediction (OPP), i.e., predicting the outcomes of a target policy using only data collected under a nominal (behavioural) policy, is a paramount problem in data-driven analysis of safety-critical systems where the deployment of a new policy may be unsafe. 
To achieve dependable off-policy predictions, recent work on {Conformal Off-Policy Prediction (COPP)} leverage the conformal prediction framework to derive prediction regions with probabilistic guarantees under the target process. Existing COPP methods can account for the distribution shifts induced by policy switching, but are limited to single-agent systems and scalar outcomes (e.g., rewards). 
In this work, we introduce \textit{MA-COPP}, the first conformal prediction method to solve OPP problems involving multi-agent systems, deriving joint prediction regions for all agents' trajectories when one or more ``ego'' agents change their policies. 
Unlike the single-agent scenario, this setting introduces higher complexity as the distribution shifts affect predictions for all agents, not just the ego agents, and the prediction task involves full multi-dimensional trajectories, not just reward values. 
A key contribution of {MA-COPP} is to avoid enumeration or exhaustive search of the output space of agent trajectories, which is instead required by existing COPP methods to construct the prediction region. We achieve this by showing that an over-approximation of the true \textit{joint prediction region} (JPR) can be constructed, without enumeration, from the maximum density ratio of the JPR trajectories. 
We evaluate the effectiveness of MA-COPP in multi-agent systems from the PettingZoo library and the \ftenth autonomous racing environment, achieving nominal coverage in higher dimensions and various shift settings. 
\end{abstract}
\section{Introduction}

In reinforcement learning, off-policy prediction (OPP)~\cite{uehara2022review} is the problem of predicting some outcome (e.g., reward, performance) of a given policy—the \textit{target policy}—using only data collected under a different policy—the \textit{behavioural policy}. Such a problem is relevant in data-driven analysis of cyber-physical systems, wherein we leverage observations/executions of the system rather than a mechanistic model, where such a model may be unavailable or just unreliable. OPP is motivated by safety-critical applications such as robotics~\cite{levine2020offline} and healthcare~\cite{murphy2001marginal}, 
where evaluating the target policy on the real system may be too risky or unethical. 

A na\"ive approach to solving the OPP problem consists of learning an (input-conditional) model of the system from behavioural data and plugging the target policy into the learned model~\cite{le2019batch}. However, such an approach does not take into account a fundamental issue of OPP: switching policies \textit{induces a distribution shift}, and so, the model inferred under the behavioural distribution cannot offer, in general, reliable predictions under the target distribution. 

In this paper, we focus on the OPP problem in systems comprising multiple interacting agents that evolve over time according to stochastic policies and stochastic dynamics.
Here, we deal with a situation where one or more \textit{ego} agents change their policies. 
This setting is considerably more challenging than the single-agent case because the distribution shift involves the predictions of all agents, not just the ego agents: 
even if the other agents remain with their behavioural (observational) policies, their actions will change in response to the shift in the ego agents' behaviours. 

To obtain reliable predictions, we leverage the framework of \textit{conformal prediction (CP)}~\cite{vovk2005algorithmic,angelopoulos2021gentle}, a technique to derive prediction regions guaranteed to cover the true (unknown) output with arbitrary probability. 
Crucially, these probabilistic guarantees are finite-sample (non-asymptotic) and distribution-free, i.e., they don't rely on priors or parametric assumptions about the data distribution. For these properties, it is unsurprising that CP has become, in recent years, the go-to method for uncertainty quantification, especially in safety-critical applications, see,~\cite{bortolussi2019neural, cairoli2023conformal, yu2023signal, yang2023safe}. 

CP uses a set of calibration points to derive a distribution of model residuals, a.k.a. \textit{scores}. The prediction region for a test point is then obtained by including all outputs whose scores appear sufficiently likely w.r.t.\ such calibration distribution. 
This procedure yields the desired probabilistic guarantees as long as the calibration and test data are exchangeable (a weaker assumption than i.i.d.). Exchangeability, however, is violated in the presence of distribution shifts, which are inherent to OPP. 
The framework of \textit{weighted exchangeability}~\cite{tibshirani2019conformal} extends CP to handle distribution shifts by reweighting the calibration points ``as if'' they were sampled under the target distribution. This is achieved by using estimates of the density ratios (DRs) between target and behavioural distributions. 
Recent works on \textit{Conformal Off-Policy Prediction (COPP)}~\cite{taufiq_conformal_2022,foffano2023conformal,zhang_conformal_2023} leverage CP and weighted exchangeability to provide valid prediction regions for OPP problems. These methods, however, only support 
single-agent systems and construct regions for scalar outcomes (e.g., reward).

In this paper, we present \textit{MA-COPP}, the first conformal prediction algorithm to solve OPP problems involving multi-agent systems. Crucially, our approach derives joint prediction regions (JPRs, akin to reach-tubes) for the future (multi-dimensional) trajectories of all agents, making it more comprehensive than existing COPP methods which are limited to scalar outcomes. 

Indeed, constructing a valid JPR under generic distribution shifts using CP and weighted exchangeability normally requires, for each test input, an exhaustive search over the output trajectory space, which is infeasible in our high-dimensional settings. To overcome this problem, we demonstrate that we can derive a conservative over-approximation of the true (unknown) JPR without resorting to exhaustive search, if the maximum density ratio (DR) over all the JPR's trajectories is known. Building on this result, in MA-COPP, we pivot the search task over the maximum DR, which is significantly more effective than an exhaustive output search.

We evaluate MA-COPP on two multi-agent case studies: a multi-particle environment from the PettingZoo library for multi-agent reinforcement learning problems~\cite{terry2021pettingzoo}, and a continuous control racing model from the F1TENTH autonomous racing environment \cite{okelly_f1tenth_2020}. We find that our approach consistently achieves (close to) target coverage under a variety of distribution shifts and for output spaces of up to $72$ dimensions, in settings where standard CP (which disregards the distribution shift) results in a coverage drop of up to $20\%$. Moreover, we find that MA-COPP yields over-conservative regions only for a very small proportion of instances.

We provide a practical example of the regions produced by the MA-COPP method using real data in Figure \ref{fig:region-comparison}, where the effect of the miscoverage error of standard CP can be clearly seen. In contrast, the MA-COPP method produces a JPR which sufficiently covers the target trajectories.

\begin{figure}
    \centering
    \includegraphics[width=0.85\linewidth]{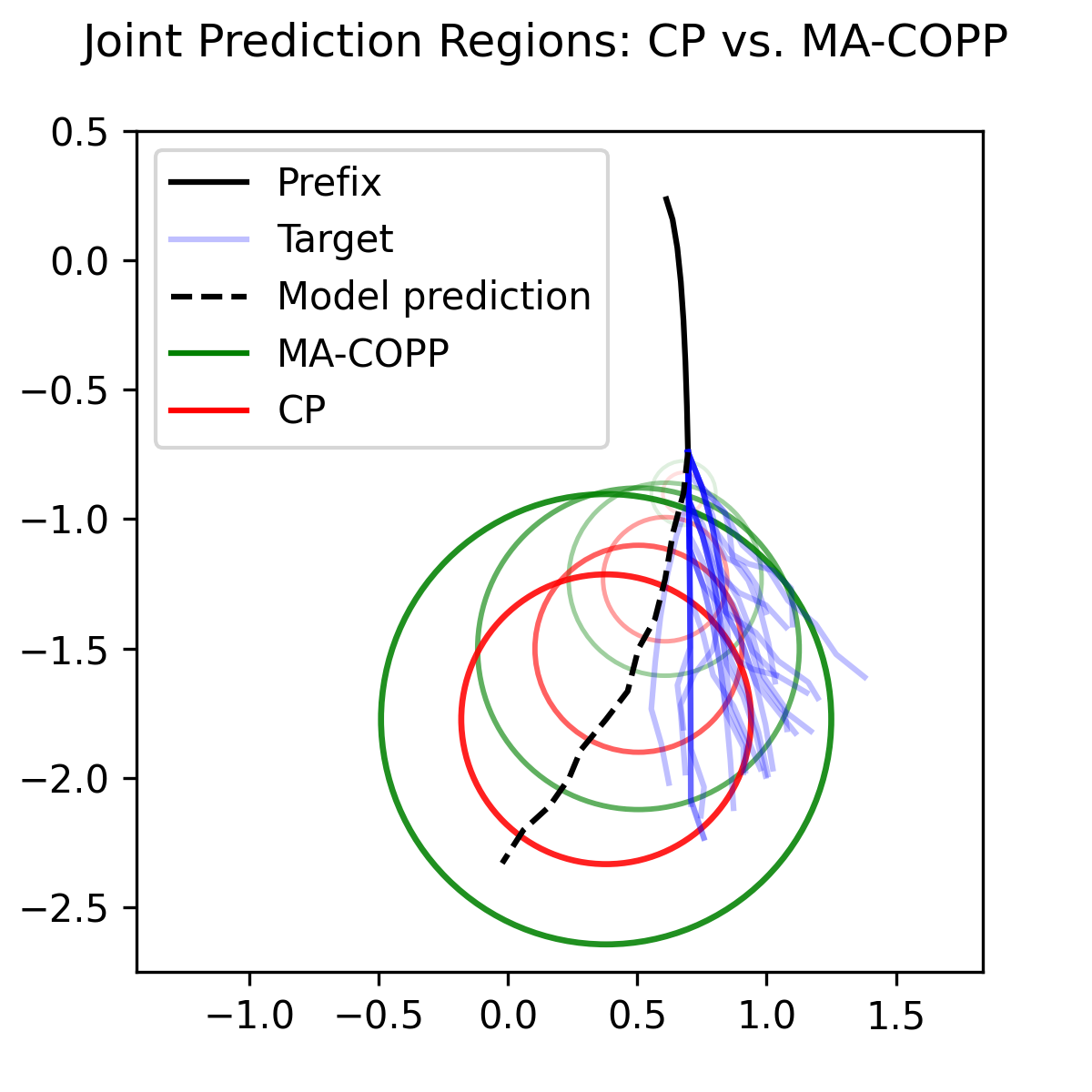}
    \caption{2D visualisation of actual JPRs constructed over the position of a single ego agent in the MPE environment defined in §~\ref{sec:mpe-results}} 
    \label{fig:region-comparison}
\end{figure}
\section{The Multi-Agent OPP Problem}

\paragraph{Notation} We use capital letters to denote random variables, lowercase for concrete values of those variables, and bold letters for the corresponding vectors (of random variables or concrete values) for all agents.  Also, we will often use the notation $\sequence{y}{1}{T}$ as a shortcut for the sequence $(y_{1},\ldots,y_{T})$.

\paragraph{Behavioural process} We consider a multi-agent system consisting of $K$ agents and described, for each agent $k=1,\ldots,K$ and time $t=1,\ldots,T$, by the following process:
\begin{equation}\label{eq:process_new}
\begin{split}
     \mathbf{X}_1 \sim & P_{init}(\cdot); \\
     A_{k,t} \sim & \pi^{b}_{k}(\cdot \mid \mathbf{X}_t); \\
      X_{k,t+1} \sim & P_k(\cdot \mid \mathbf{X}_t, A_{k,t})
\end{split}
\end{equation}

\noindent where $X_{k,t}\in \mathbb{R}^{n}$ is the agent's state, $\mathbf{X}_t=(X_{k,t})_{k=1}^K \in \mathbb{R}^{n\times K}$ is the environment/global state;  $A_{k,t} \in \mathcal{A}_k$ is the performed action, and $\mathcal{A}_k$ is the discrete or continuous action space of the agent; 
$P_{init}: \mathbb{R}^{n\times K} \rightarrow [0,1]$ is the distribution of initial environment states;
$\pi_{k}^{b}:(\mathbb{R}^{n\times K} \times \mathcal{A}_k) \rightarrow [0,1]$ is the stochastic behavioural policy of the agent; 
and $P_k: (\mathbb{R}^{n\times K} \times \mathcal{A}_k \times \mathbb{R}^{n})\rightarrow [0,1]$ is the transition probability function determining the distribution of the next agent's states.
  
Each agent has access to the global environment state but does not know the actions of the other agents. Furthermore, we assume a subset of \textit{ego agents} $E \subseteq \{1,\ldots, K\}$ as the agents that are under our direct control and allowed to switch policies in the target process.

In our OPP settings, we do not have access to the process defined in~\eqref{eq:process_new}, that is, we do not know the initial probabilities $P_{init}$, the transition probabilities $P_{k}$ and the non-ego behavioural policies $\{\pi_{k}^{b}\}_{k\not\in E}$. We only know the ego behavioural policy $\{\pi_{e}^{b}\}_{e\in E}$ and have access to the observational data defined below. For simplicity, we will now refer to terms pertaining to ego agents with subscript $e$ and the non-ego agents with subscript $k$.

\paragraph{Observational data} We can observe a set of $N$ global state trajectories of length $T$, and for ego agents only, we can also observe the actions they perform at each step of the trajectory. More formally, we have access to the following dataset of 
realisations of the process $\left(\mathbf{X}_1, \left(A_{e,1}\right)_{e\in E}, \mathbf{X}_2, \left(A_{e,2}\right)_{e\in E}, \ldots,\mathbf{X}_T\right)$ induced by~\eqref{eq:process_new}:

\begin{equation}\label{eq:trajectories_new}
    \mathcal{D}=\left\{ \left(  \mathbf{x}^{(i)}_{t},  \left(a^{(i)}_{e,t}\right)_{e\in E} \right)_{t=1}^T \right\}_{i=1}^N.
\end{equation}

\paragraph{Problem definition} 

We focus on the problem of deriving a \textit{joint prediction region} for the future agent trajectories $\sequence{\mathbf{X}}{H+1}{T}$, given a prefix $\sequence{\mathbf{X}}{1}{H}$, where $H$ is the prefix length and $T>H$ is the total trajectory length. 
In particular, this problem is one of \textit{off-policy prediction}: using observational data only (i.e., realisations of the behavioural process), we seek to construct a prediction region for $\sequence{\mathbf{X}}{H+1}{T}$ under a target policy different from the behavioural one. We consider the case where the policies of the ego agents $e\in E$ change after time $H$ and the policies of all the other agents remain fixed, described by the \textit{target process} in \eqref{eq:process_target}, where only $A_{k,t}$ differs from the behavioural process.

\begin{equation}\label{eq:process_target}
\begin{split}
    & A_{k,t} \sim \pi_{k,t}(\cdot \mid \mathbf{X}_t);
\end{split}
\end{equation}

where all non-ego agents always follow their behavioural policy ($\pi_{k,t} = \pi^b_k$ if $k\notin E$) and each ego agent initially follows their behavioural policy ($\pi_{k,t} = \pi^b_k$ if $t\leq H$) and switches their policy to a \emph{known} target policy thereafter ($\pi_{e,t} = \pi^*_e$ if $t>H$). 
Below, we denote with $P^{\pi^b}$ and $P^{\pi^*}$ the distributions  under the behavioural process~\eqref{eq:process_new} and the target process~\eqref{eq:process_target}, respectively.

\begin{problem}[Multi-Agent Off-Policy Prediction]\label{prob:ma-copp}

Given observations $\mathcal{D}\sim_{iid} P^{\pi^b}$ of length $T$ and of the form~\eqref{eq:trajectories_new}
generated by the behavioural process~\eqref{eq:process_new}, construct
a \emph{joint prediction region (JPR)} $C_{\alpha}(\cdot)$ with coverage $1-\alpha$ for i.i.d.\ test trajectories $\mathbf{X}_{1\ldots T}\sim P^{\pi^*}$, where $\alpha$ is the desired miscoverage rate.
Formally, the JPR must satisfy the following:

\begin{equation}
\mathbb{P}_{\mathbf{X}_{1\ldots T}\sim P^{\pi^*}}\left(\mathbf{X}_{H+1\ldots T}\in C_{\alpha}(\mathbf{X}_{1\ldots H})\right)\geq 1-\alpha.
\end{equation}

where $\mathbf{X}_{1\ldots T}\sim P^{\pi^*}$ is generated according to the target process~\eqref{eq:process_target}.
\end{problem}

\begin{remark}
    Although process~\eqref{eq:process_new} is Markovian, our predictions consider, in input, a sequence of past states to accommodate models (e.g., autoregressive, recurrent) that make use of a sequence of past states.
\end{remark}

\section{Background}\label{sec:background}

\noindent\emph{Conformal Prediction:} an uncertainty quantification framework that can be applied on top of any supervised learning task for constructing distribution-free prediction regions with guaranteed marginal coverage.  
We now introduce the method using a standard regression example. Starting from  a dataset  $\mathcal{D} = \{ (x^{(i)}, y^{(i)}) \}_i$ of $(x, y)$ pairs sampled i.i.d.\ from an unknown distribution $P_{X,Y}$. 
CP performs the following steps:

\begin{enumerate}
    \item Split $\mathcal{D}$ into disjoint training and calibration datasets, $\mathcal{D}_{t}$ and $\mathcal{D}_{c}$, with $|\mathcal{D}_{c}|$
    being the number of samples in $\mathcal{D}_{c}$;
    \item Train a predictor $\hat{T}:X\rightarrow Y$ using $\mathcal{D}_{t}$;
    \item Define a \textit{non-conformity score function} $S: X \times Y \rightarrow \mathbb{R}$, such that $S(x,y)$ quantifies the discrepancy/residual between $y$ and the prediction $\hat{T}(x)$;
    \item Use $\mathcal{D}_{c}$ to define the calibration distribution
    \begin{equation}
        \widehat{F} = \sum_{i=1}^{|\mathcal{D}_{c}|} \dfrac{1}{|\mathcal{D}_{c}|+1} \delta_{s_i} + \dfrac{1}{|\mathcal{D}_{c}|+1}\delta_{\infty},
    \end{equation}
    where $\delta_{s}$ is the Dirac distribution with parameter $s$,  $s_i=S(x^{(i)}, y^{(i)})$ is the score of the $i$-th calibration point, and $\delta_{\infty}$ represents the unknown score of the test point\footnote{Since we do not know the true output for the test input, a worst-case score of $\infty$ is assumed.};
    \item For a given test point $x$ and failure rate $\alpha$, construct the prediction region as $C_{\alpha}(x)=\{y: S(x, y)\le Q_{1-\alpha}(\widehat{F})\}$, where $Q_{1-\alpha}(\widehat{F})$ is the $1-\alpha$ quantile of $\widehat{F}$. Such a prediction region satisfies the following coverage guarantee w.r.t.\ unseen test data $(x,y)\sim P_{X,Y}$:
    \begin{equation}\label{eq:standard_cp}
        \mathbb{P}_{(x,y)\sim P_{X,Y}}(y\in C_{\alpha}(x)) \ge 1-\alpha.
    \end{equation}
    Note that the above holds in the more general case when $(x,y)$ is exchangeable w.r.t.\ calibration data, i.e., when the joint probability of $(x^{1},y^{1}),\ldots,(x^{|\mathcal{D}_{c}|},y^{|\mathcal{D}_{c}|}),(x,y)$ remains the same for any permutation of the data points.
\end{enumerate} 

\begin{remark}\label{rem:implicit_region}
In standard CP for regression, to construct $C_{\alpha}(x)$ it is not required to enumerate and check inclusion for all candidate outputs. For instance, a common choice for $S$ is $S(x,y)=\| y - \hat{T}(x) \|_p$ for some $p\geq 1$, and so $C_{\alpha}(x)$ can be constructed implicitly (without enumeration) as the $L_p$-ball centered at $\hat{T}(x)$ with radius $Q_{1-\alpha}(\widehat{F})$. This is possible because the quantile $Q_{1-\alpha}(\widehat{F})$ is the same for all candidate outputs. 
\end{remark}

\noindent\emph{CP under Distribution Shift:} Standard CP provides marginal coverage guarantees in the case where data observed at test-time is sampled from an exchangeable distribution as that of the calibration set $\mathcal{D}_c$.
Should the distribution of test data differ in this regard (i.e. $\mathcal{D}_{test} \sim P^*_{X,Y}\neq P_{X,Y}$), it induces a distribution shift which violates the exchangeability assumption and with that, the coverage guarantee~\eqref{eq:standard_cp}. The weighted exchangeability notion of~\cite{tibshirani2019conformal} extends CP to deal with such shifts. It does so by reweighting the probabilities of the calibration distribution $\widehat{F}$ by the \textit{density ratio (DR)} $w(x,y) = \mathrm{d}P^*_{X,Y}(x,y)/\mathrm{d}P_{X,Y}(x,y)$: in this way, we transform $\widehat{F}$ as if the scores had been computed over the target distribution. 
For a test point $x$ and candidate output $y$, the reweighted distribution becomes:
\begin{equation}\label{eq:weighted_calib_dist}
\widehat{F}(x,y) = \sum_{i=1}^{|\mathcal{D}_{c}|} \dfrac{w(x^{(i)}, y^{(i)})}{W + w(x,y)} \delta_{s_i}+ \dfrac{w(x,y)}{W + w(x,y)}\delta_{\infty} 
\end{equation}
where $W=\sum_{i=1}^{|\mathcal{D}_{c}|} w(x^{(i)}, y^{(i)})$. We denote with $p_i(x,y)$ the probability of the $i$-th calibration point reweighted as above. Note that each score $s_i$ has a higher (lower) probability if the target distribution makes $(x^{(i)}, y^{(i)})$ more (less) likely. 
By using $\widehat{F}(x,y)$ to construct $C_{\alpha}(x)$, we recover guarantee~\eqref{eq:standard_cp} for when  $(x,y)\sim P^*_{X,Y}$\cite[Theorem~2]{tibshirani2019conformal}.

\begin{remark}\label{rem:enumeration}
    The reweighted calibration distribution depends on the test point $(x,y)$ because the probability $p_{|{\mathcal{D}_c|}+1}$ of the test score needs to be reweighted by $w(x,y)$. This implies that, to construct a prediction region $C_{\alpha}(x)$ for a given test input $x$, we need to reweight $\widehat{F}$ for every candidate output $y$ to determine if $y\in C_{\alpha}(x)$ by checking $S(x,y)\leq Q_{1-\alpha}(\widehat{F}(x,y))$. Thus, for general shifts, $ C_{\alpha}(x)$ needs to be constructed by \emph{enumerating} (and checking) individual candidates $y$, because the quantile $Q_{1-\alpha}(\widehat{F}(x,y))$ changes with $y$, i.e., the region cannot be constructed implicitly as per Remark~\ref{rem:implicit_region}. This is not the case, however, with covariate shift, whereby $P_X$ changes but $P_{Y\mid X}$ does not: the DR now depends only on the input, and so $\widehat{F}(x,y)$ remains the same for every $y$.
\end{remark}

\noindent\emph{Conformal Off-Policy Prediction:}
The work of~\cite{taufiq_conformal_2022} builds on weighted exchangeability for conformal off-policy prediction (COPP) in contextual bandits—a special (and simpler) case of our problem, restricted to one step (i.e., $T=2$ and $H=1$), one agent ($K=1$), and scalar outcomes (as opposed to our multi-dimensional JPRs). 
In this setting, changing the behavioural policy $\pi^{b}$ into the target policy $\pi^{*}$ induces a shift in the distribution of $Y \mid X$, while the distribution of $X$ stays the same. Hence, the DR can be derived as follows:

\begin{equation}\label{eq:copp_DR}
\begin{split}
    w(x,y) &= 
    \dfrac{\mathrm{d}P^{\pi^{*}}_{X,Y}(x,y)}{\mathrm{d}P^{\pi^{b}}_{X,Y}(x,y)} = \dfrac{\mathrm{d}P^{\pi^{*}}_{Y\mid X}(x,y) \cdot \mathrm{d}P^{\pi^{*}}_{X}(x,y)}{\mathrm{d}P^{\pi^{b}}_{Y\mid X}(x,y)\cdot \mathrm{d}P^{\pi^{b}}_{X}(x,y)} \\
    &= \dfrac{\mathrm{d}P^{\pi^{*}}_{Y\mid X}(x,y)}{\mathrm{d}P^{\pi^{b}}_{Y\mid X}(x,y)}
    = \dfrac{\int P(y \mid a, x) \cdot \pi^*(a \mid x) \mathrm{d}a}{\int P(y \mid a, x) \cdot \pi^b(a \mid x) \mathrm{d}a}
\end{split}
\end{equation}

Since the transition probabilities $P$ are unknown, the COPP approach derives an estimation $\hat{w}$ of $w$ by plugging in~\eqref{eq:copp_DR}  a data-driven surrogate $\hat{P}$ learned from (behavioural) data and by approximating the integrals with Monte-Carlo sampling of the policies.  
To construct $C_{\alpha}(x)$, instead of enumerating all candidate outputs $y$ (required as per Remark~\ref{rem:enumeration}), COPP implements an exhaustive grid search over the output space and returns the interval closure of all $y$s that pass the CP inclusion test. This operation is costly but remains feasible because, in the COPP method, $y$ is a scalar. 

Using an estimation of the DR has limitations, in that it directly affects the accuracy of the reweighted distribution~\eqref{eq:weighted_calib_dist} and consequently the validity guarantees. In \cite{lei2021conformal}, the authors show that the miscoverage gap induced by using $\hat{w}$ instead of $w$ is bounded by $\frac{1}{2}\mathbb{E}_{X,Y\sim P^{\pi^{b}}_{X,Y}}|w(X,Y)-\hat{w}(X,Y)|$. 
However, the deviation $|w(X,Y)-\hat{w}(X,Y)|$ and the resulting miscoverage error are inevitably exacerbated when considering, like in our settings, sequential processes involving multiple steps. 
\section{The \frameworkname\ approach}

The \frameworkname\ method provides a solution to Problem~\ref{prob:ma-copp} by constructing valid joint prediction regions (JPRs) for the agents' future trajectories under the target policy despite only having access to observational data under the behavioural policy. To do so, we extend weighted exchangeability and the COPP method to deal with high-dimensional output spaces and JPRs arising from multiple steps and multiple agents. 

An overview of the algorithm is presented in Figure~\ref{fig:method-overview}.

\begin{figure*}
    \centering
    \vspace{10pt}
    \includegraphics[width=.85\textwidth]{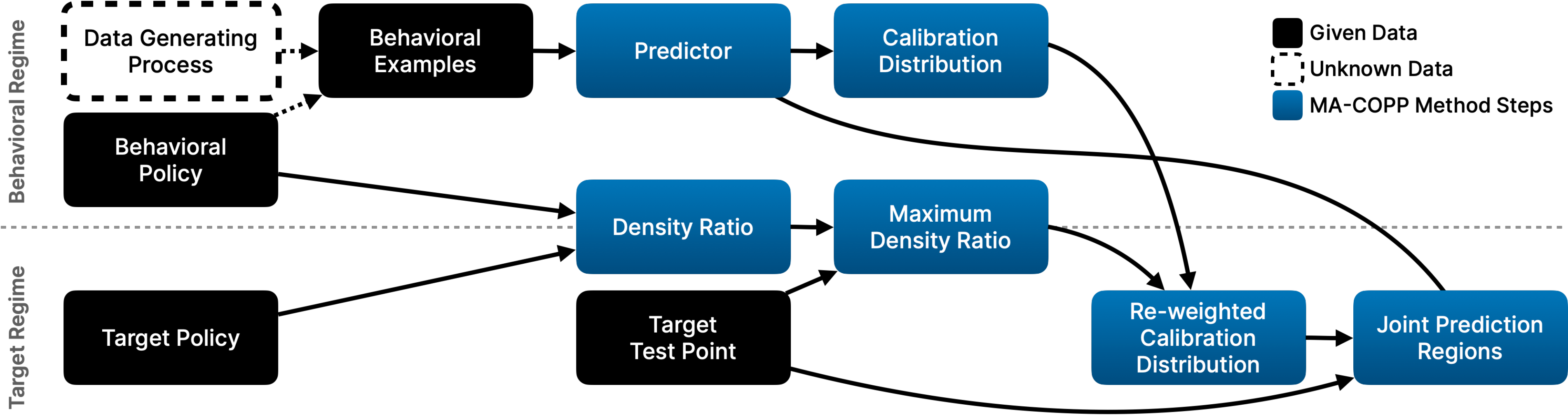}
    \caption{Overview of the MA-COPP method. A calibration distribution is first derived from behavioural data and a predictor (see §~\ref{sec:non-conformity-score}). Density ratios are computed as described in §~\ref{sec:density-ratio}. To construct the JPR for a given test point, we estimate the maximum DR over all the outputs that pass the CP test, and use this estimate to reweight the calibration distribution, see §~\ref{sec:construct-pred-regions}.}
    \label{fig:method-overview}
    \vspace{-10pt}
\end{figure*}

\paragraph{Lifting of prediction task} 
We start by slightly reformulating our prediction task. Instead of constructing a JPR $C_{\alpha}$ for sequences of future global states $\sequence{\mathbf{X}}{H+1}{T}$, as per Problem~\ref{prob:ma-copp}, we will construct a JPR $C^+_{\alpha}$  for sequences of future global states \textit{and} ego agents actions $\left(\mathbf{X}_{t}, \left(A_{t,e}\right)_{e \in E}\right)_{t=H+1}^T$. For simplicity, we denote this ``augmented'' sequence by $\sequence{\mathbf{Y}}{H+1}{T}$. 
As explained below, we can trivially use a JPR that is valid for $\sequence{\mathbf{Y}}{H+1}{T}$ to construct a JPR for $\sequence{\mathbf{X}}{H+1}{T}$ (i.e., for the original problem).

\begin{proposition}
For $\alpha\in (0,1)$ and $H<T$, let $C^+_{\alpha}$ be a  JPR valid for $\sequence{\mathbf{Y}}{H+1}{T}$, i.e., such that
$$\mathbb{P}_{(\mathbf{X}_{1\ldots H},\mathbf{Y}_{H+1\ldots T})\sim P^{\pi^*}}\left(\mathbf{Y}_{H+1\ldots T}\in C^+_{\alpha}(\mathbf{X}_{1\ldots H})\right)\geq 1-\alpha.$$ 
Let $C'_{\alpha}(\sequence{\mathbf{X}}{1}{H}) = \Pi_{\sequence{\mathbf{x}}{H+1}{T}}( C^+_{\alpha}(\sequence{\mathbf{X}}{1}{H}))$ be the projection of the JPR onto components $\sequence{\mathbf{x}}{H+1}{T}$. Then, $C'_{\alpha}$ provides a solution to Problem~\ref{prob:ma-copp} in that
$$
\mathbb{P}_{(\sequence{\mathbf{X}}{1}{H},\sequence{\mathbf{X}}{H+1}{T})\sim P^{\pi^*}}\left(\sequence{\mathbf{X}}{H+1}{T}\in C'_{\alpha}(\sequence{\mathbf{X}}{1}{H}) \right)\geq 1-\alpha.
$$
\end{proposition}

\begin{proof}
For simplicity, we denote $\mathbf{X}_{1\ldots H}$ by $X$, $\mathbf{X}_{H+1\ldots T}$ by $X'$, and $\mathbf{A}_{H+1\ldots T,e}$ by $A$. So, we can define the coverage of $(X,X',A)$ (the augmented sequence) w.r.t.\ the prediction region $C_{\alpha}^+(X)$ by $\mathbb{E}_{X,X',A}[f(X,X',A)]$, where $f$ is the indicator function telling if $(X',A) \in C_{\alpha}^+(X)$. By the premises of the proposition, we have $\mathbb{E}_{X,X',A}[f(X,X',A)]\geq 1-\alpha$. So, we need to prove that $\mathbb{E}_{X,X'}[g(X,X')]\geq \mathbb{E}_{X,X',A}[f(X,X',A)]$ where $g$ is the indicator function telling if $X' \in \Pi_{x'}(C_{\alpha}^+(X))$. We have the following derivation:

\vspace{-.5em}
\[
\small
\begin{array}{l}
\mathbb{E}_{X,X',A}[f(X,X',A)]\\
= {\displaystyle \int f(x,x',a) \cdot \mathbb{P}_{X,X',A}(x,x',a) \mathrm{d}x,\mathrm{d}x',\mathrm{d}a}\\
\leq  {\displaystyle \int g(x,x') \cdot \mathbb{P}_{X,X',A}(x,x',a) \mathrm{d}x,\mathrm{d}x',\mathrm{d}a}\\
= {\displaystyle \int g(x,x') \cdot \left(\int \mathbb{P}_{X,X',A}(x,x',a) \mathrm{d}a\right) \mathrm{d}x,\mathrm{d}x'}\\
= {\displaystyle \int g(x,x') \cdot \mathbb{P}_{X,X'}(x,x') \mathrm{d}x,\mathrm{d}x'}
=\mathbb{E}_{X,X'}[g(X,X')],
\end{array}
\]
where the first inequality holds because for any $x$,$x'$, and $a$, we have that $f(x,x',a)\leq g(x,x')$ (because, by definition, $(x',a)\in  C_{\alpha}^+(x) \rightarrow x' \in \Pi_{x'}(C_{\alpha}^+(x))$).
\end{proof}

Moving our focus to a more complex prediction task may seem counter-intuitive, but it allows us to use a precise definition of DR, as we will see next.

\paragraph{Density ratio computation} \label{sec:density-ratio}
Similarly to COPP, our change of policy causes a shift in the distribution of $\sequence{\mathbf{Y}}{H+1}{T} \mid \sequence{\mathbf{X}}{1}{H}$, while the distribution of $\sequence{\mathbf{X}}{1}{H}$ stays the same. Hence, we have the following derivation for the DR:
\begin{equation}\label{eq:ma-copp-dr}
\begin{split}
    & w(\sequence{\mathbf{x}}{1}{H},\sequence{\mathbf{y}}{H+1}{T}) = \dfrac{\mathrm{d}P^{\pi^{*}}_{\sequence{\mathbf{Y}}{H+1}{T} \mid \sequence{\mathbf{X}}{1}{H}}(\sequence{\mathbf{x}}{1}{H},\sequence{\mathbf{y}}{H+1}{T})}{\mathrm{d}P^{\pi^{b}}_{\sequence{\mathbf{Y}}{H+1}{T} \mid \sequence{\mathbf{X}}{1}{H}}(\sequence{\mathbf{x}}{1}{H},\sequence{\mathbf{y}}{H+1}{T})} \\
    &= \left(\dfrac{\prod_{t=H}^{T-1} \prod_{e \in E} P_e(x_{e,t+1}\mid \mathbf{x}_t,a_{e,t})\pi^*_{e}(a_{e,t}\mid \mathbf{x}_t)}{\prod_{t=H}^{T-1} \prod_{e \in E} P_e(x_{e,t+1}\mid \mathbf{x}_t,a_{e,t})\pi^b_{e}(a_{e,t}\mid \mathbf{x}_t)} \right. \times \\
    & \qquad\quad \left. \dfrac{\prod_{t=H}^{T-1}\prod_{k \not\in E} \int P_k(x_{k,t+1}\mid \mathbf{x}_t,a_k)\pi^b_{k}(a_k\mid \mathbf{x}_t) \mathrm{d}a_k}{\prod_{t=H}^{T-1}\prod_{k \not\in E} \int P_k(x_{k,t+1}\mid \mathbf{x}_t,a_k)\pi^b_{k}(a_k\mid \mathbf{x}_t) \mathrm{d}a_k} \right) \\
    &= \dfrac{\prod_{t=H}^{T-1} \prod_{e \in E} \pi^*_{e}(a_{e,t}\mid \mathbf{x}_t) }{\prod_{t=H}^{T-1} \prod_{e \in E} \pi^b_{e}(a_{e,t}\mid \mathbf{x}_t)}
\end{split}
\end{equation}

Owing to the reformulation of the prediction task, we can now express the DR in terms of the ego agents' policies only—which are known. This means that we can compute the DR precisely. Without such a reformulation, we would need to estimate the agent actions and transition probabilities jointly, as in \eqref{eq:copp_DR}, leading to an approximate DR.

Note that, for $w$ to be well-defined, the likelihood of $(\sequence{\mathbf{x}}{1}{H},\sequence{\mathbf{y}}{H+1}{T})$ cannot be zero in both behavioural and target distributions. This assumption holds, for instance, when for any state $\mathbf{x}$, $\pi_e^b(\cdot \mid \mathbf{x})$ has full support on $\mathcal{A}_e$ for any ego agent $e\in E$, and $P_k(\cdot \mid \mathbf{x},a_{k})$
has full support on $\mathbb{R}^n$ for any agent $k$ and action $a_k\in \mathcal{A}_k$\footnote{Full support can be ensured, for instance, by adding Gaussian noise to the outputs of $\pi_e^b$ and $P_k$.}. 

\paragraph{Non-conformity score} \label{sec:non-conformity-score}
To define the score function $S(\sequence{\mathbf{x}}{1}{H},\sequence{\mathbf{y}}{H+1}{T})$, we need to establish a predictor $\hat{T}$ first. In our case, $\hat{T}$ is a multivariate regression function $\hat{T}: \sequence{\mathbf{x}}{1}{H} \mapsto \sequence{\hat{\mathbf{x}}}{H+1}{T}$ mapping a state sequence $\sequence{\mathbf{x}}{1}{H}$ into an estimate of its continuation $\sequence{\hat{\mathbf{x}}}{H+1}{T}$\footnote{Note that we do not require $\hat{T}$ to predict actions but only states. Our final aim is to construct JPRs for state sequences and so actions should not affect the score $S(\sequence{\mathbf{x}}{1}{H},\sequence{\mathbf{y}}{H+1}{T})$, that is, the deviation between predictions and ground truth. Actions are only involved in the reweighting of $\hat{F}$ (through the DR~\eqref{eq:ma-copp-dr}).}. 
We choose the score function recently proposed in~\cite{cleaveland2023conformal}, which is designed for time series prediction tasks. This score describes the deviation between two (multi-dimensional) trajectories in terms of the maximum deviation across all time points. In particular, for an arbitrary choice of $\gamma_{H+1},\ldots,\gamma_{T}>0$, $S$ is defined as follows:
\begin{equation}\label{eq:non-conf-score}
\begin{split}
    & S(\sequence{\mathbf{x}}{1}{H}, \sequence{\mathbf{y}}{H+1}{T}) =\\
    & \max_{t=H+1 \ldots T}\left\{\gamma_t \left\| \sigma \circ \left( \hat{T}(\sequence{\mathbf{x}}{1}{H}) - \sequence{\mathbf{y}}{H+1}{T} \right) \right\|_2 \right\},
\end{split}
\end{equation}

\noindent where $\hat{T}(\sequence{\mathbf{x}}{1}{H})$ (resp., $\sequence{\mathbf{y}}{H+1}{T}$) is the global state at time $t$ according to the prediction (resp., the output $\sequence{\mathbf{y}}{H+1}{T}$), and $\sigma$ includes normalisation constants for each dimension. 
The $\gamma_t$ parameters determine how much the residuals at different time points contribute to the score. Since residuals tend to grow as the prediction horizon increases, without such parameters (i.e., if  $\gamma_t=1$ for all $t$), the score would be dominated by prediction errors at time points far ahead in the horizon. Thus, a sensible choice for $\gamma_t$ is any monotonically decreasing series -- in our experiments we choose $\gamma_t=(t-H)^{-1}$ -- which mitigates the above issue by assigning more importance to prediction errors in the immediate future. 
Finally, given a test input $\sequence{\mathbf{x}}{1}{H}$, we define the JPR $C^+_{\alpha}$ as follows:

\begin{equation}\label{eq:pred_reg_new}
\begin{split}
	C^+_{\alpha}(\sequence{\mathbf{x}}{1}{H}) = \{ &\sequence{\mathbf{y}}{H+1}{T} \mid S(\sequence{\mathbf{x}}{1}{H},\sequence{\mathbf{y}}{H+1}{T}) \\
    & \quad \leq Q_{1-\alpha}(\hat{F}(\sequence{\mathbf{x}}{1}{H},\sequence{\mathbf{y}}{H+1}{T})) \}
\end{split}
\end{equation}

where $\hat{F}(\sequence{\mathbf{x}}{1}{H},\sequence{\mathbf{y}}{H+1}{T})$ is the calibration distribution constructed with behavioural data and reweighted according to the DR $w(\sequence{\mathbf{x}}{1}{H},\sequence{\mathbf{y}}{H+1}{T})$ (see Eq.~\ref{eq:weighted_calib_dist}).

\paragraph{Construction of Prediction Regions}\label{sec:construct-pred-regions} As discussed in Remark~\ref{rem:enumeration}, with distribution shifts (other than covariate shifts), the calibration distribution needs to be reweighted for every candidate output. 
However, enumerating the outputs as in \cite{taufiq_conformal_2022} quickly becomes infeasible as the size of the search space is exponential in the number of output dimensions.

To alleviate this problem, various search strategies could be applied. Then, the prediction region would be constructed by taking some closure (to avoid zero-volume regions) of those visited trajectories that pass the CP test. However, such strategies are incomplete, and the resulting region will necessarily undercover.

Crucially, our MA-COPP approach overcomes the issue of enumerating or searching the output space, based on the following intuition.
Due to the above limitations, we cannot compute precisely $C^+_{\alpha}(\sequence{\mathbf{x}}{1}{H})$, i.e., the true JPR~\eqref{eq:pred_reg_new}. 
However, if we knew the value $w^\top(C^+_{\alpha}(\sequence{\mathbf{x}}{1}{H}))$ of the maximum DR among all trajectories in $C^+_{\alpha}(\sequence{\mathbf{x}}{1}{H})$, then 
we could use the same $w^\top(C^+_{\alpha}(\sequence{\mathbf{x}}{1}{H}))$ for all candidate outputs $\sequence{\mathbf{y}}{H+1}{T}$ when reweighting the calibration distribution. The resulting region would meet the coverage guarantees because it is a conservative approximation of $C^+_{\alpha}(\sequence{\mathbf{x}}{1}{H})$ and, importantly, can be constructed implicitly without enumerating the output space because all outputs have now the same critical value, as per Remark~\ref{rem:implicit_region}. 
Before introducing this statement more formally, let us denote with $\widehat{F}(w)$ the calibration distribution $\widehat{F}$ reweighted by $w \in (0,\infty)$. 

\begin{proposition}[Max-DR region]\label{prop:maxdr}
Let $\mathcal{Y} \subseteq C^+_{\alpha}(\sequence{\mathbf{x}}{1}{H})$ be a subset of the JPR~\eqref{eq:pred_reg_new}. Let $w^{\top}(\mathcal{Y}) = \underset{\sequence{\mathbf{y}}{H+1}{T} \in \mathcal{Y}}{\max} w(\sequence{\mathbf{x}}{1}{H}, \sequence{\mathbf{y}}{H+1}{T})$ be the maximum DR within $\mathcal{Y}$. Define the \emph{max-DR region} as $C^+_{\alpha}(\sequence{\mathbf{x}}{1}{H}, w^{\top}(\mathcal{Y}))= \{ \sequence{\mathbf{y}}{H+1}{T} \mid S(\sequence{\mathbf{x}}{1}{H},\sequence{\mathbf{y}}{H+1}{T}) \leq Q_{1-\alpha}(\hat{F}(w^{\top}(\mathcal{Y}))) \}$. Then, $\mathcal{Y} \subseteq C^+_{\alpha}(\sequence{\mathbf{x}}{1}{H}, w^{\top}(\mathcal{Y}))$.
\end{proposition}

\begin{proof}
    To prove that $\mathcal{Y} \subseteq C^+_{\alpha}(\sequence{\mathbf{x}}{1}{H}, w^{\top}(\mathcal{Y}))$, we show that for any $\sequence{\mathbf{y}}{H+1}{T} \in \mathcal{Y}$, $Q_{1-\alpha}(\hat{F}(\sequence{\mathbf{x}}{1}{H},\sequence{\mathbf{y}}{H+1}{T})) \leq Q_{1-\alpha}(\hat{F}(w^{\top}(\mathcal{Y})))$. This follows from~\eqref{eq:weighted_calib_dist}, since 
    the reweighted probability of the test point (which has $\infty$ score) is always higher in $\hat{F}(w^{\top}(\mathcal{Y}))$, i.e., $\dfrac{w^{\top}(\mathcal{Y})}{W + w^{\top}(\mathcal{Y})} \geq \dfrac{w(\sequence{\mathbf{x}}{1}{H},\sequence{\mathbf{y}}{H+1}{T})}{W + w(\sequence{\mathbf{x}}{1}{H},\sequence{\mathbf{y}}{H+1}{T})}$ for any $\sequence{\mathbf{y}}{H+1}{T} \in \mathcal{Y}$.
\end{proof}

A consequence of Proposition~\ref{prop:maxdr} is that, given the true maximum DR  $w^{\top}(C^+_{\alpha}(\sequence{\mathbf{x}}{1}{H}))$, the corresponding max-DR region is valid (i.e., it has coverage of at least $1-\alpha$) because it contains the true JPR $C^+_{\alpha}(\sequence{\mathbf{x}}{1}{H})$. However, $w^{\top}(C^+_{\alpha}(\sequence{\mathbf{x}}{1}{H}))$ is not known a priori and needs to be estimated. 

Our approach uses search techniques to find an under-approximation $\tilde{w}^\top\leq w^{\top}(C^+_{\alpha}(\sequence{\mathbf{x}}{1}{H}))$ of the true maximum DR. While the resulting max-DR region  $C^+_{\alpha}(\sequence{\mathbf{x}}{1}{H},\tilde{w}^\top)$ may not achieve the target coverage, pivoting the search task over the maximum DR is substantially more effective than doing so over the output space directly -- and this is confirmed by our experiments. The reason is that with our approach, as soon as we find \textit{just one} output trajectory $\sequence{\mathbf{y}}{H+1}{T}$ of the true JPR (i.e., that passes the CP test), then we can use the corresponding DR $w(\sequence{\mathbf{x}}{1}{H},\sequence{\mathbf{y}}{H+1}{T})$ to construct a max-DR region that is guaranteed to include \textit{all} trajectories of the true JPR with a DR below or equal to $w(\sequence{\mathbf{x}}{1}{H},\sequence{\mathbf{y}}{H+1}{T})$ (see Prop.~\ref{prop:max-DR-2} below). On the contrary, without our approach, we would need an extensive (and likely infeasible) search to cover {all} those trajectories sufficiently well. 

\begin{proposition}\label{prop:max-DR-2}
    For $\tilde{w}^\top\in (0,\infty)$, the corresponding max-DR region $C^+_{\alpha}(\sequence{\mathbf{x}}{1}{H}, \tilde{w}^\top)$ contains the JPR subset $\{\sequence{\mathbf{y}}{H+1}{T} \in  C^+_{\alpha}(\sequence{\mathbf{x}}{1}{H}) \mid w(\sequence{\mathbf{x}}{1}{H},\sequence{\mathbf{y}}{H+1}{T})\leq \tilde{w}^\top\}$.
\end{proposition}

\begin{proof}
    For any $\sequence{\mathbf{y}}{H+1}{T} \in  C^+_{\alpha}(\sequence{\mathbf{x}}{1}{H})$, it holds that $S(\sequence{\mathbf{x}}{H+1}{T},\sequence{\mathbf{y}}{H+1}{T}) \leq Q_{1-\alpha}(\hat{F}(w(\sequence{\mathbf{x}}{1}{H},\sequence{\mathbf{y}}{H+1}{T})))$. As shown in the proof of Prop.~\ref{prop:maxdr}, if $w(\sequence{\mathbf{x}}{1}{H},\sequence{\mathbf{y}}{1}{H})\leq \tilde{w}^\top$, then we have $S(\sequence{\mathbf{x}}{H+1}{T},\sequence{\mathbf{y}}{H+1}{T}) \leq Q_{1-\alpha}(\hat{F}(\tilde{w}^\top))$, i.e., $\sequence{\mathbf{y}}{H+1}{T} \in C^+_{\alpha}(\sequence{\mathbf{x}}{1}{H}, \tilde{w}^\top)$. 
\end{proof}

To estimate the maximum DR, our search technique performs sampling of a \textit{synthetic target process}, an approximation of the target process where the unknown transition probabilities $P_k$ and the non-ego policies ${\pi}^{b}_{k}$ are replaced by data-driven approximations $\hat{P}_k$ and $\hat{\pi}^{b}_{k}$ learned from behavioural data\footnote{For efficiency, one may choose to learn an end-to-end approximation of the non-ego dynamics, i.e., a single predictor mapping $\mathbf{X}_t$ directly into $(X_{k,t+1})_{k\not\in E}$ (thus avoiding to learn the non-ego agents' policies). We define $\hat{\pi}^{b}_{k}$ and $\hat{P}_k$ as isotropic multi-variate Gaussians with parameters predicted by a neural network model.}.

\section{Results}

\subsection{Experimental Settings}
We evaluate our method on two case studies, a collaborative environment with a continuous state space and a discrete action space, and an adversarial environment with a continuous state and action space.
In both case studies, we compare the MA-COPP method against two configurations. 1) $T \rightarrow T$ is a standard CP approach using the true target distribution for both calibration and test data: the gold standard approach where everything is known. 2) $B \rightarrow T$ is standard CP under distribution shift whereby we use the behavioural distribution for the calibration data, and the target distribution at test-time: the standard approach which disregards distribution shifts. For all experiments, we consider the nominal coverage rate $\alpha = 0.95$

\subsection{Case Studies}

\subsubsection{Multi-Particle Environment (MPE)}\label{sec:mpe-results}

We use an MPE environment based on the PettingZoo \cite{terry2021pettingzoo} library. In this collaborative 2D environment with discrete actions, there are $k$ agents and $m$ {\em landmarks}. Landmarks take the form of static circles in the state space. In our experiments, we set $k = m = 3$, with only one ego agent. The goal of the environment is for the agents to cooperatively cover all of the landmarks whilst avoiding collisions with each other.

The state space of the environment is given by the vector:

\begin{equation}\label{eq:state_space}
    \mathbf{x}_{t} = \left( \left( x_{k,t}, v_{k,t} \right)^{K}_{k=1}, l_{1}, \ldots, l_{M} \right)
\end{equation}
where $x_{k,t}$ and $v_{k,t}$ are the position and velocity vectors for each agent at time $t$. Each agent can only observe the position (with additive noise sampled from a Gaussian distribution) of the other agents and so its observation space is defined as follows:
\begin{equation}\label{eq:obs_space}
    y_{k,t} = \left( x_{k,t}, v_{k,t}, \left( x_{j,t} \right)_{j\neq k}, l_{1}, \ldots, l_{M} \right) + w_{k}^{sensor}
\end{equation}

At each time step, the agents sample their actions from a discrete space $\mathcal{A} =$ \textit{\{left, right, up, down, do nothing\}} according to their own stochastic policies, which are parameterized by neural networks. Each action updates the agent's velocity by a constant vector—one for each of the cardinal directions w.r.t. the origin, as well as $[0, 0]$. The transition kernel, $P$, for each agent can be defined as follows:
\begin{equation}\label{eq:motion}
    X_{k,t+1} = X_{k,t} + V_{k,t} + W_k^{act}; \ V_{k,t+1} = V_{k,t} + U_{k,t}
\end{equation}
where $X_k$ and $V_k$ are the position and velocity of agent $k$ and $W_k^{act}$ is a Gaussian noise term modelling actuation noise. The resultant control input $U_{k,t}$ that an agent receives is the unit vector corresponding to the cardinal direction (i.e. $[0,1]^{T}$ if $A_{k,t} = \text{up}$) multiplied a scalar representing the intensity of the acceleration.

\begin{figure*}[t!]
    \centering
    \begin{subfigure}[b]{0.24\textwidth}
        \centering
        \includegraphics[height=1.5in]{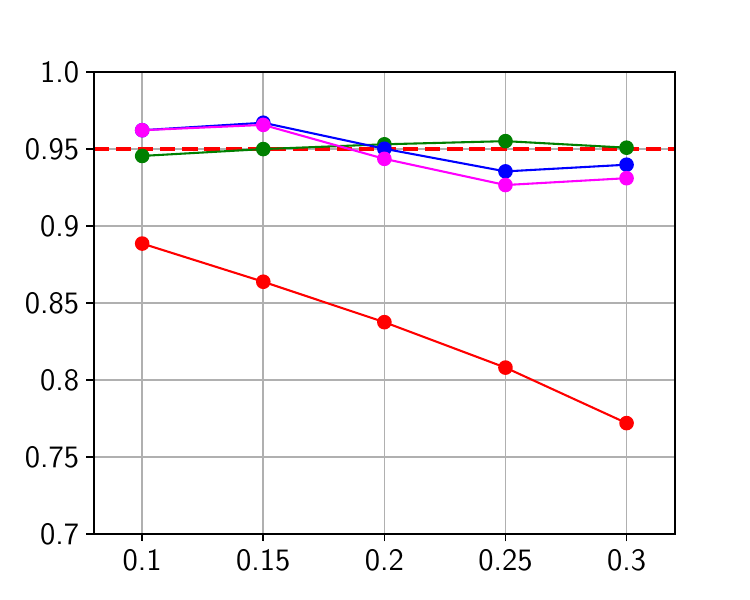}
        \caption{Marginal coverage rates with a prediction length of $8$ timesteps}
        \label{fig:mpe_marginal_8}
    \end{subfigure}%
    ~
    \begin{subfigure}[b]{0.24\textwidth}
        \centering
        \includegraphics[height=1.5in]{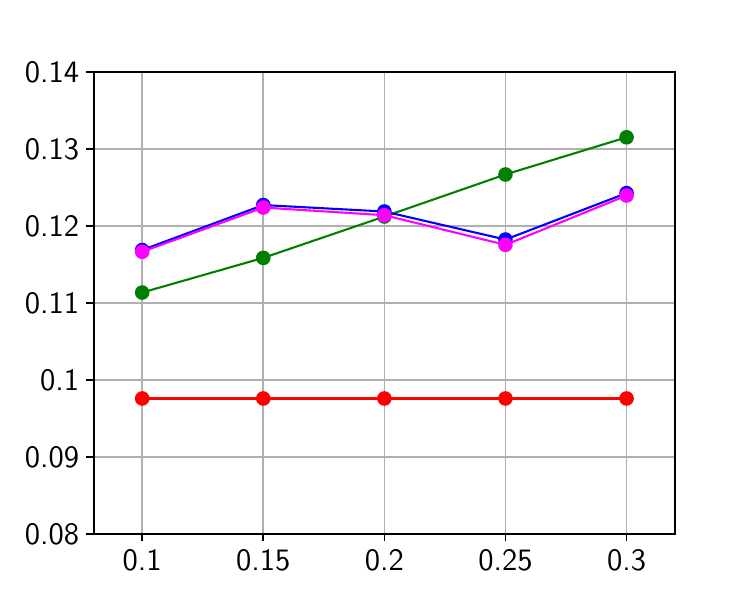}
        \caption{Critical values of JPRs with a prediction length of $8$ timesteps}
        \label{fig:mpe_critical_8}
    \end{subfigure}%
    ~
    \begin{subfigure}[b]{0.24\textwidth}
        \centering
        \includegraphics[height=1.5in]{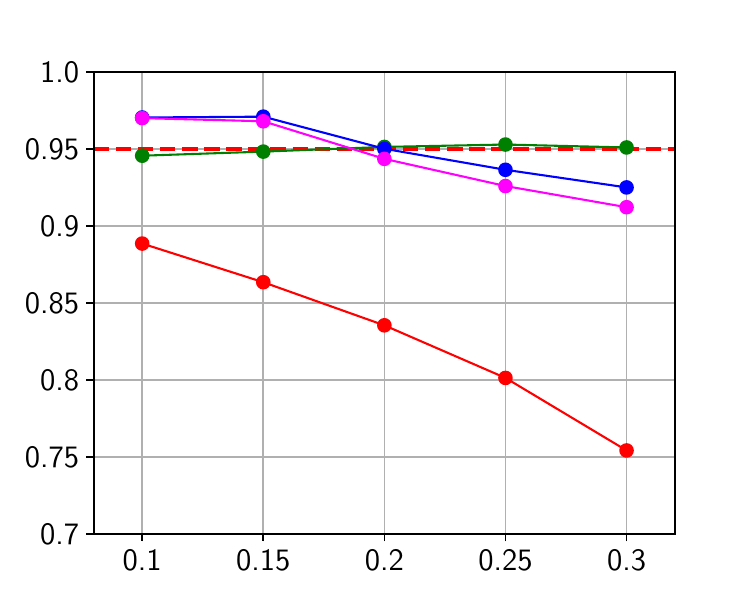}
        \caption{Marginal coverage rates with a prediction length of $12$ timesteps}
        \label{fig:mpe_marginal_12}
    \end{subfigure}%
    ~ 
    \begin{subfigure}[b]{0.24\textwidth}
        \centering
        \includegraphics[height=1.5in]{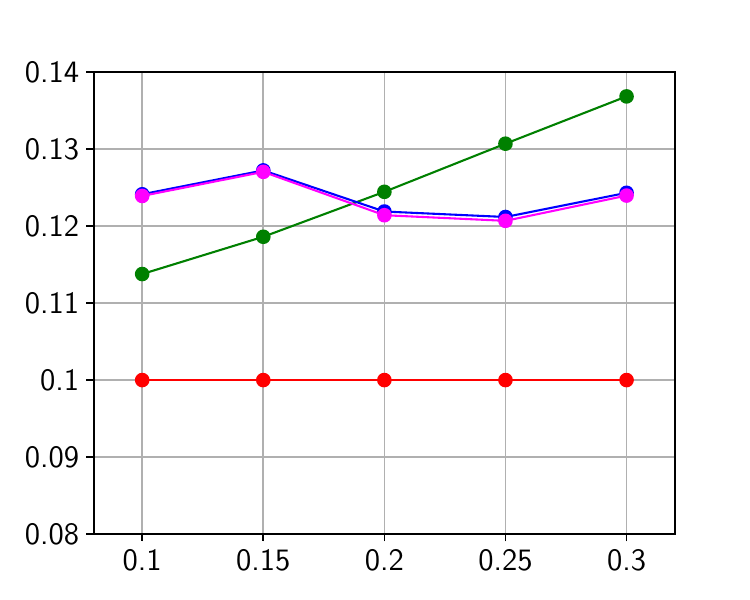}
        \caption{Critical values of JPRs with a prediction length of $12$ timesteps}
        \label{fig:mpe_critical_12}%
    \end{subfigure}%
    \caption{$x$-axis: $\epsilon_{bias}$; Dashed red: nominal coverage level; Red: standard CP $B \rightarrow T$; Green: standard CP $T \rightarrow T$; Blue: MA-COPP with synthetic process $B^{S} \rightarrow T$; Magenta: MA-COPP with true process $B^{T} \rightarrow T$}
\end{figure*}

We define the each behavioural policy as $\epsilon$-greedy, that is, with probability $1-\epsilon_{greedy}$ it selects the true action $A_{k,t}^{(i)}$ or a random action $A_{k,t}^{(j)}, i \neq j$ with probability $\epsilon_{greedy} / |\mathcal{A}| - 1$. The target policies for the ego agents are defined in the same way, with the exception of an additional \textit{bias} term which we use to control the degree of distribution shift. This bias modifies the action probabilities further by selecting a fixed action (in our case, \textit{down}) with probability $\epsilon_{bias}$.  In our experiments, we fix $\epsilon_{greedy}^{b} = 0.1$ and $\epsilon_{greedy}^{*}$, and evaluate our method with various $\epsilon_{bias} \in \{0.1, 0.15, 0.2, 0.25, 0.3\}$.

We generate 3,200 prefixes (1,600--800--800 for training--calibration--test respectively). The training set is used only for learning an LSTM, as required in \ref{sec:non-conformity-score} as well as learning the dynamics models for the synthetic process. For each prefix, we generate 25 Monte-Carlo continuations, both under the behavioural and true target—which we assume to be unknown, but use only for the sake of evaluation. We also sample 25 continuations from the synthetic process. For all of the datasets, we generate $9$ steps and $12$ steps for the prefix and continuations respectively. The prediction regions are formed over the continuation of test prefixes, and are of the dimension $k \times 2 \times (T-H+1)$.

\begin{table}[b]
\begin{tabular}{|c|c|c|c|c|c|}
\hline
\multicolumn{1}{|l|}{\backslashbox{$T$}{$\epsilon$-bias}} & 0.1                       & 0.15                      & 0.2                          & 0.25                         & 0.3                          \\ \hline
8                         & 0                         & 0                         & 0.06625                      & 0.0475                       & 0.0575                       \\ \hline
12                        & \multicolumn{1}{l|}{0.01} & \multicolumn{1}{l|}{0.05} & \multicolumn{1}{l|}{0.06625} & \multicolumn{1}{l|}{0.06375} & \multicolumn{1}{l|}{0.05125} \\ \hline
\end{tabular}\caption{Proportion of test points with unbounded prediction regions}
\end{table}\label{tab:mpe_infinity}
 
We find that the MA-COPP method performs very well with respect to standard CP, and results in useful and tight prediction regions. When evaluating marginal coverage rates in particular, we observe in both figures \ref{fig:mpe_critical_8} and \ref{fig:mpe_marginal_12}, that with a relatively low epsilon bias of $0.1$, the coverage of standard CP breaks down whilst the MA-COPP method continues to provide the 95\% coverage guarantee until $\epsilon = 0.2$. Beyond this point, MA-COPP also experiences to drop in coverage below 95\%. This is caused by an under-approximation of the true max-DR, leading to prediction regions that are too tight and therefore under-cover, albeit to a far lesser degree than standard CP.

In addition to the coverage results, we also observe a corresponding increase in region size as the JPRs grow to allow for more uncertainty. Since we assume that the true process is unknown and must rely on our synthetic process for the search over the max-DR (as described in \ref{sec:construct-pred-regions}), we also evaluate a fourth setting in which we perform the reweighting with the true target distribution ($B^{T} \rightarrow T$). Crucially, we observe that in both figures \ref{fig:mpe_critical_8} and \ref{fig:mpe_critical_12}, we see that using the synthetic data-generating process (instead of the true process) proves to be sufficient for exploring the max-DR value necessary to increase the critical value sufficiently to cover the true target output space. 

As the degree of policy shift increases, the max-DR often becomes very large. This results in test points where the critical value after reweighting is $\infty$. These regions are unavoidable in cases with large shift, however, as shown in Table \ref{tab:mpe_infinity}, only a small proportion of test points have such regions in our MPE experiments.

\subsubsection{F1TENTH}

We also evaluated MA-COPP in a competitive, continuous control racing environment, presented in \cite{okelly_f1tenth_2020}. In the \ftenth environment, agents are in control of racecars, with states $s = \left(x,y,\delta,v,\Psi,\dot{\Psi},\beta \right)\in \mathbb{R}^7$. $x,y$ represent the $x$ and $y$ position of the vehicle center. $\delta$ represents the steering angle, and $v$ the forward velocity of the car. $\Psi, \dot{\Psi}$ denote the vehicle's yaw angle and yaw rate, and $\beta$ gives the vehicle's slip angle at its centre. The actions of these vehicles are also continuous, $a = (v, \delta) \in\mathbb{R}^2$. $v \in [0,6]$ represents the  velocity, and $\delta \in [-0.4,0.4]$ the steering angle. The full model is available in~\cite{althoff_commonroad_2017}. The simulator is run at a frequency of 100Hz.

Three agents, one ego, and two adversaries compete in a race beginning at various locations along a track. The opponent agents operate at $95\%$ of the speed of the ego agent and follow pre-specified racelines that are centerline offsets. The ego agent's nominal trajectory will bring it into collision with the opponents. A collision avoidance algorithm is used to avoid this. We define the composition of these two components as our nominal policy. The behavioural policy is an $\epsilon^b_{greedy}=0.2$ policy as described in Section \ref{sec:mpe-results}. The target policies replicate this approach with probabilities $\epsilon_{greedy}^t = \left\{0.3,0.4,0.5,0.6,0.7\right\}$, representing a shift away from the nominal policy.
In the \ftenth experiments, the data-generating process and the synthetic process used to search for the maximum DR are the same.

We predict agent trajectories four timesteps into the future. $5,077$ prefixes ($2538$--$1270$--$1269$ for training--calibration--test respectively) were generated. For each of the $5,077$ prefixes, $25$ behavioural and target suffixes, as well as $50$ synthetic suffixes each were generated. The prediction is done for each agent position, resulting in predictions with dimension $(3,2,4)$, for 3 agents over 4 timesteps.

Results are presented in Table~\ref{tab:f110-results}. First, we note that standard CP ($B \rightarrow T$) performs worse as $\epsilon$ grows, losing more than $2\%$ coverage at $\epsilon=0.7$. We expect this behaviour as the target distribution shifts away from the behavioural. In contrast, the \frameworkname~ approach compensates for this shift and consistently provides coverage in line with the nominal $95\%$ target. We also examine the gold-standard CP approach ($T \rightarrow T$), which gives us a benchmark for the critical value required to achieve the desired coverage. In all examples, we see that the \frameworkname~ approach has a larger critical value and, therefore, larger regions than the $T \rightarrow T$ approach. While MA-COPP does result in larger regions, it also provides coverage closer to the target coverage. We can see that in the continuous control example presented, MA-COPP succeeds in compensating for policy shift when standard CP approaches experience coverage gaps and that it does so without very conservative regions. When the distribution shift is too great, MA-COPP can provide trivial guarantees, as the density ratio can quickly explode. These results have been omitted for space.

\begin{table}[t]
    \centering
    \begin{tabular}{|c|c|c|c|c|}
    \hline
        \multicolumn{2}{|c|}{Approach} & \multicolumn{1}{c|}{\multirow{2}{*}{T $\rightarrow$ T}} & \multicolumn{1}{c|}{\multirow{2}{*}{B $\rightarrow$ T}} & \multicolumn{1}{c|}{\multirow{2}{*}{MA-COPP}} \\ \cline{1-2}
        \multicolumn{1}{|c|}{$\epsilon$} & Metric                 & & & \\ \hline
        \multirow{2}{*}{0.3}    & Coverage  & 94.22\%   & 94.26\%   & \textbf{95.02\%} \\ \cline{2-5} 
                                & Avg. CV   & 1.689     & 1.692     & 1.742 \\ \hline
        \multirow{2}{*}{0.4}    & Coverage  & 94.45\%   & 94.32\%   & \textbf{94.94\%} \\ \cline{2-5} 
                                & Avg. CV   & 1.701     & 1.692     & 1.738 \\ \hline
        \multirow{2}{*}{0.5}    & Coverage  & 94.24\%   & 93.92\%   & \textbf{94.78\%} \\ \cline{2-5} 
                                & Avg. CV   & 1.714     & 1.692     & 1.754 \\ \hline
        \multirow{2}{*}{0.6}    & Coverage  & 94.39\%   & 93.79\%   & \textbf{95.23\%} \\ \cline{2-5} 
                                & Avg. CV   & 1.731     & 1.692     & 1.796 \\ \hline
        \multirow{2}{*}{0.7}    & Coverage  & 94.16\%   & 92.99\%   & \textbf{95.51\%} \\ \cline{2-5} 
                                & Avg. CV   & 1.766     & 1.692     & 1.867 \\ \hline
    \end{tabular}
    \caption{Results from the \ftenth experiments, with a suffix length of 4 timesteps. The notation X $\rightarrow$ Y is used as shorthand for standard split conformal prediction, fit on the precedent, and evaluated on the consequent.}
    \label{tab:f110-results}
\end{table}

\section{Related Work} 

Conformal prediction (CP)~\cite{vovk2005algorithmic, angelopoulos2021gentle} is a popular uncertainty quantification method, with numerous applications including 
computer vision~\cite{angelopoulos2022image}, language models~\cite{quach2023conformal}, and system verification~\cite{bortolussi2021neural, cairoli2021neural,lindemann2023conformal, cairoli2023conformal}. Recently, conformal prediction has also been used for safe robotic planning and control, see, e.g.,~\cite{lindemann2023safe, dixit2023adaptive, yang2023safe, ren2023robots, yu2023signal, muthali2023multi}.  Our work is also related to time-series forecasting, for which we find a growing body of CP-based approaches such as~\cite{stankeviciute2021conformal,sun2022copula}.

Past COPP work and our work leverage the weighted exchangeability method of~\cite{tibshirani2019conformal} (discussed in more detail in Section~\ref{sec:background}), but other CP-based approaches exist that address the distribution shift problem, such as \textit{adaptive CP}~\cite{gibbs2021adaptive}, an approach that dynamically adjusts the coverage level in an online manner to compensate for observed under/over-coverage under unknown distribution shifts, or 
\textit{robust CP}~\cite{gendler2021adversarially}, which extends CP to handle bounded adversarial input perturbations.

While our work is the first to support multiple agents and prediction of multi-dimensional trajectories, it was inspired by prior COPP methods for single-agent systems such as~\cite{taufiq_conformal_2022} (see Section~\ref{sec:background}), the work of~\cite{foffano2023conformal} which considers dynamic models (Markov Decision Processes) but predicts scalar outcomes (the value of the MDP trajectories), and the subsampling-based approach of \cite{zhang_conformal_2023} which only supports discrete actions and becomes prohibitive for long trajectories. 

\section{Conclusion}

We presented MA-COPP, the first conformal prediction method for reliable off-policy prediction in multi-agent systems. Our approach avoids the output space enumeration that frustrates existing COPP approaches by reweighting (for every test input) the calibration distribution only once, using an estimate of the maximum density ratio. 
We evaluated our method on two case studies, respectively involving discrete and continuous action spaces, demonstrating that MA-COPP succeeds in adjusting the coverage without generating excessively conservative regions, even in cases where standard CP massively undercovers. 
\section{Disclosure of Funding}

This work is supported by the \textit{REXASI-PRO} H-EU project, call HORIZON-CL4-2021-HUMAN-
01-01, grant agreement ID: 101070028; the
Engineering and Physical Sciences Research Council (EPSRC) under Award EP/W014785/2; NSF CCRI 1925587; NSF GRFP award DGE-2236662; and US DoT Safety21 National University Transportation Center.

\bibliographystyle{IEEEtran}
\bibliography{references}

% Generated by IEEEtran.bst, version: 1.14 (2015/08/26)
\begin{thebibliography}{10}
\providecommand{\url}[1]{#1}
\csname url@samestyle\endcsname
\providecommand{\newblock}{\relax}
\providecommand{\bibinfo}[2]{#2}
\providecommand{\BIBentrySTDinterwordspacing}{\spaceskip=0pt\relax}
\providecommand{\BIBentryALTinterwordstretchfactor}{4}
\providecommand{\BIBentryALTinterwordspacing}{\spaceskip=\fontdimen2\font plus
\BIBentryALTinterwordstretchfactor\fontdimen3\font minus
  \fontdimen4\font\relax}
\providecommand{\BIBforeignlanguage}[2]{{%
\expandafter\ifx\csname l@#1\endcsname\relax
\typeout{** WARNING: IEEEtran.bst: No hyphenation pattern has been}%
\typeout{** loaded for the language `#1'. Using the pattern for}%
\typeout{** the default language instead.}%
\else
\language=\csname l@#1\endcsname
\fi
#2}}
\providecommand{\BIBdecl}{\relax}
\BIBdecl

\bibitem{uehara2022review}
M.~Uehara, C.~Shi, and N.~Kallus, ``A review of off-policy evaluation in
  reinforcement learning,'' \emph{arXiv preprint arXiv:2212.06355}, 2022.

\bibitem{levine2020offline}
S.~Levine, A.~Kumar, G.~Tucker, and J.~Fu, ``Offline reinforcement learning:
  Tutorial, review, and perspectives on open problems,'' \emph{arXiv preprint
  arXiv:2005.01643}, 2020.

\bibitem{murphy2001marginal}
S.~A. Murphy, M.~J. van~der Laan, J.~M. Robins, and C.~P. P.~R. Group,
  ``Marginal mean models for dynamic regimes,'' \emph{Journal of the American
  Statistical Association}, vol.~96, no. 456, pp. 1410--1423, 2001.

\bibitem{le2019batch}
H.~Le, C.~Voloshin, and Y.~Yue, ``Batch policy learning under constraints,'' in
  \emph{International Conference on Machine Learning}.\hskip 1em plus 0.5em
  minus 0.4em\relax PMLR, 2019, pp. 3703--3712.

\bibitem{vovk2005algorithmic}
V.~Vovk, A.~Gammerman, and G.~Shafer, \emph{Algorithmic learning in a random
  world}.\hskip 1em plus 0.5em minus 0.4em\relax Springer, 2005, vol.~29.

\bibitem{angelopoulos2021gentle}
A.~N. Angelopoulos and S.~Bates, ``A gentle introduction to conformal
  prediction and distribution-free uncertainty quantification,'' \emph{arXiv
  preprint arXiv:2107.07511}, 2021.

\bibitem{bortolussi2019neural}
L.~Bortolussi, F.~Cairoli, N.~Paoletti, S.~A. Smolka, and S.~D. Stoller,
  ``Neural predictive monitoring,'' in \emph{Runtime Verification: 19th
  International Conference, RV 2019, Porto, Portugal, October 8--11, 2019,
  Proceedings 19}.\hskip 1em plus 0.5em minus 0.4em\relax Springer, 2019, pp.
  129--147.

\bibitem{cairoli2023conformal}
F.~Cairoli, N.~Paoletti, and L.~Bortolussi, ``Conformal quantitative predictive
  monitoring of stl requirements for stochastic processes,'' in
  \emph{Proceedings of the 26th ACM International Conference on Hybrid Systems:
  Computation and Control}, 2023, pp. 1--11.

\bibitem{yu2023signal}
X.~Yu, Y.~Zhao, X.~Yin, and L.~Lindemann, ``Signal temporal logic control
  synthesis among uncontrollable dynamic agents with conformal prediction,''
  \emph{arXiv preprint arXiv:2312.04242}, 2023.

\bibitem{yang2023safe}
S.~Yang, G.~J. Pappas, R.~Mangharam, and L.~Lindemann, ``Safe perception-based
  control under stochastic sensor uncertainty using conformal prediction,''
  \emph{arXiv preprint arXiv:2304.00194}, 2023.

\bibitem{tibshirani2019conformal}
R.~J. Tibshirani, R.~Foygel~Barber, E.~Candes, and A.~Ramdas, ``Conformal
  prediction under covariate shift,'' \emph{Advances in neural information
  processing systems}, vol.~32, 2019.

\bibitem{taufiq_conformal_2022}
\BIBentryALTinterwordspacing
M.~F. Taufiq, J.-F. Ton, R.~Cornish, Y.~W. Teh, and A.~Doucet, ``Conformal
  off-policy prediction in contextual bandits,'' in \emph{Advances in Neural
  Information Processing Systems}, S.~Koyejo, S.~Mohamed, A.~Agarwal,
  D.~Belgrave, K.~Cho, and A.~Oh, Eds., vol.~35.\hskip 1em plus 0.5em minus
  0.4em\relax Curran Associates, Inc., 2022, pp. 31\,512--31\,524. [Online].
  Available:
  \url{https://proceedings.neurips.cc/paper_files/paper/2022/file/cc84bfabe6389d8883fc2071c848f62a-Paper-Conference.pdf}
\BIBentrySTDinterwordspacing

\bibitem{foffano2023conformal}
D.~Foffano, A.~Russo, and A.~Proutiere, ``Conformal off-policy evaluation in
  markov decision processes,'' \emph{arXiv preprint arXiv:2304.02574}, 2023.

\bibitem{zhang_conformal_2023}
Y.~Zhang, C.~Shi, and S.~Luo, ``Conformal off-policy prediction,'' in
  \emph{International Conference on Artificial Intelligence and
  Statistics}.\hskip 1em plus 0.5em minus 0.4em\relax PMLR, 2023, pp.
  2751--2768.

\bibitem{terry2021pettingzoo}
J.~Terry, B.~Black, N.~Grammel, M.~Jayakumar, A.~Hari, R.~Sullivan, L.~S.
  Santos, C.~Dieffendahl, C.~Horsch, R.~Perez-Vicente \emph{et~al.},
  ``Pettingzoo: Gym for multi-agent reinforcement learning,'' \emph{Advances in
  Neural Information Processing Systems}, vol.~34, pp. 15\,032--15\,043, 2021.

\bibitem{okelly_f1tenth_2020}
\BIBentryALTinterwordspacing
M.~O’Kelly, H.~Zheng, D.~Karthik, and R.~Mangharam,
  ``\BIBforeignlanguage{en}{{F1TENTH}: {An} {Open}-source {Evaluation}
  {Environment} for {Continuous} {Control} and {Reinforcement} {Learning}},''
  in \emph{\BIBforeignlanguage{en}{Proceedings of the {NeurIPS} 2019
  {Competition} and {Demonstration} {Track}}}.\hskip 1em plus 0.5em minus
  0.4em\relax PMLR, Aug. 2020, pp. 77--89, iSSN: 2640-3498. [Online].
  Available: \url{https://proceedings.mlr.press/v123/o-kelly20a.html}
\BIBentrySTDinterwordspacing

\bibitem{lei2021conformal}
L.~Lei and E.~J. Cand{\`e}s, ``Conformal inference of counterfactuals and
  individual treatment effects,'' \emph{Journal of the Royal Statistical
  Society Series B: Statistical Methodology}, vol.~83, no.~5, pp. 911--938,
  2021.

\bibitem{cleaveland2023conformal}
M.~Cleaveland, I.~Lee, G.~J. Pappas, and L.~Lindemann, ``Conformal prediction
  regions for time series using linear complementarity programming,''
  \emph{arXiv preprint arXiv:2304.01075}, 2023.

\bibitem{althoff_commonroad_2017}
\BIBentryALTinterwordspacing
M.~Althoff, M.~Koschi, and S.~Manzinger, ``{CommonRoad}: {Composable}
  benchmarks for motion planning on roads,'' in \emph{2017 {IEEE} {Intelligent}
  {Vehicles} {Symposium} ({IV})}, Jun. 2017, pp. 719--726. [Online]. Available:
  \url{https://ieeexplore.ieee.org/document/7995802}
\BIBentrySTDinterwordspacing

\bibitem{angelopoulos2022image}
A.~N. Angelopoulos, A.~P. Kohli, S.~Bates, M.~Jordan, J.~Malik, T.~Alshaabi,
  S.~Upadhyayula, and Y.~Romano, ``Image-to-image regression with
  distribution-free uncertainty quantification and applications in imaging,''
  in \emph{International Conference on Machine Learning}.\hskip 1em plus 0.5em
  minus 0.4em\relax PMLR, 2022, pp. 717--730.

\bibitem{quach2023conformal}
V.~Quach, A.~Fisch, T.~Schuster, A.~Yala, J.~H. Sohn, T.~S. Jaakkola, and
  R.~Barzilay, ``Conformal language modeling,'' \emph{arXiv preprint
  arXiv:2306.10193}, 2023.

\bibitem{bortolussi2021neural}
L.~Bortolussi, F.~Cairoli, N.~Paoletti, S.~A. Smolka, and S.~D. Stoller,
  ``Neural predictive monitoring and a comparison of frequentist and bayesian
  approaches,'' \emph{International Journal on Software Tools for Technology
  Transfer}, vol.~23, no.~4, pp. 615--640, 2021.

\bibitem{cairoli2021neural}
F.~Cairoli, L.~Bortolussi, and N.~Paoletti, ``Neural predictive monitoring
  under partial observability,'' in \emph{Runtime Verification: 21st
  International Conference, RV 2021, Virtual Event, October 11--14, 2021,
  Proceedings 21}.\hskip 1em plus 0.5em minus 0.4em\relax Springer, 2021, pp.
  121--141.

\bibitem{lindemann2023conformal}
L.~Lindemann, X.~Qin, J.~V. Deshmukh, and G.~J. Pappas, ``Conformal prediction
  for stl runtime verification,'' in \emph{Proceedings of the ACM/IEEE 14th
  International Conference on Cyber-Physical Systems (with CPS-IoT Week 2023)},
  2023, pp. 142--153.

\bibitem{lindemann2023safe}
L.~Lindemann, M.~Cleaveland, G.~Shim, and G.~J. Pappas, ``Safe planning in
  dynamic environments using conformal prediction,'' \emph{IEEE Robotics and
  Automation Letters}, 2023.

\bibitem{dixit2023adaptive}
A.~Dixit, L.~Lindemann, S.~X. Wei, M.~Cleaveland, G.~J. Pappas, and J.~W.
  Burdick, ``Adaptive conformal prediction for motion planning among dynamic
  agents,'' in \emph{Learning for Dynamics and Control Conference}.\hskip 1em
  plus 0.5em minus 0.4em\relax PMLR, 2023, pp. 300--314.

\bibitem{ren2023robots}
A.~Z. Ren, A.~Dixit, A.~Bodrova, S.~Singh, S.~Tu, N.~Brown, P.~Xu, L.~Takayama,
  F.~Xia, J.~Varley \emph{et~al.}, ``Robots that ask for help: Uncertainty
  alignment for large language model planners,'' \emph{arXiv preprint
  arXiv:2307.01928}, 2023.

\bibitem{muthali2023multi}
A.~Muthali, H.~Shen, S.~Deglurkar, M.~H. Lim, R.~Roelofs, A.~Faust, and
  C.~Tomlin, ``Multi-agent reachability calibration with conformal
  prediction,'' \emph{arXiv preprint arXiv:2304.00432}, 2023.

\bibitem{stankeviciute2021conformal}
K.~Stankeviciute, A.~M~Alaa, and M.~van~der Schaar, ``Conformal time-series
  forecasting,'' \emph{Advances in neural information processing systems},
  vol.~34, pp. 6216--6228, 2021.

\bibitem{sun2022copula}
S.~Sun and R.~Yu, ``Copula conformal prediction for multi-step time series
  forecasting,'' \emph{arXiv preprint arXiv:2212.03281}, 2022.

\bibitem{gibbs2021adaptive}
I.~Gibbs and E.~Candes, ``Adaptive conformal inference under distribution
  shift,'' \emph{Advances in Neural Information Processing Systems}, vol.~34,
  pp. 1660--1672, 2021.

\bibitem{gendler2021adversarially}
A.~Gendler, T.-W. Weng, L.~Daniel, and Y.~Romano, ``Adversarially robust
  conformal prediction,'' in \emph{International Conference on Learning
  Representations}, 2021.

\end{thebibliography}

\end{document}